\documentclass[11pt]{amsart}
\usepackage{amsmath}
\usepackage{color}
\usepackage[margin=1in]{geometry}
\date{\today}

\renewcommand{\span}{\mathop{\rm span}}
\newcommand{\be}{\begin{equation}}
\newcommand{\ee}{\end{equation}}
\newcommand{\bea}{\begin{eqnarray}}
\newcommand{\eea}{\end{eqnarray}}
\newcommand{\eq}[1]{(\ref{#1})}

\newtheorem{thm}{Theorem}[section]
\newtheorem{prop}[thm]{Proposition}

\newtheorem{lem}[thm]{Lemma}

\renewcommand{\epsilon}{\varepsilon}

 \def\idty{{\mathchoice {\mathrm{1\mskip-4mu l}} {\mathrm{1\mskip-4mu l}} %
{\mathrm{1\mskip-4.5mu l}} {\mathrm{1\mskip-5mu l}}}}

\begin{document}

\title[On the Dynamics of Lattice Systems]{On the dynamics of lattice systems\\ with unbounded on-site terms in the Hamiltonian}

\author[B. Nachtergaele]{Bruno Nachtergaele$^1$}
\thanks{B.\ N.\ was supported in part by the National Science Foundation under grant DMS-1009502}
\address{$^1$ Department of Mathematics\\ University of California, Davis\\ Davis, CA 95616, USA}
\email{bxn@math.ucdavis.edu}

\author[R. Sims]{Robert Sims$^2$}
\thanks{R.\ S.\ was supported in part by the National Science Foundation under grant DMS-1101345 and by the
Simons Foundation under grant $\#$301127}
\address{$^2$ Department of Mathematics\\
University of Arizona\\
Tucson, AZ 85721, USA}
\email{rsims@math.arizona.edu}

\date{}
\begin{abstract}
We supply the mathematical arguments required to complete the proofs of two previously published results: Lieb-Robinson bounds for the dynamics of quantum lattice systems with unbounded on-site terms in the Hamiltonian and the existence of the thermodynamic limit of the dynamics of such systems.
\end{abstract}
\maketitle

%
%
%

\section{Introduction}
In \cite{nach2009} we studied Lieb-Robinson bounds for lattice oscillator systems and applied those bounds to prove the existence of infinite system dynamics for a class of anharmonic lattices in \cite{nach2010}. This required a generalization of the Lieb-Robinson bounds proved for short-range quantum spin systems in \cite{LR1} and the extensions to more general interactions in \cite{HK, NS1, NOS1}. In particular, oscillator systems have an infinite-dimensional Hilbert space of states at each site and their Hamiltonians are unbounded self-adjoint operators, defined on a dense domain, not the entire Hilbert space. Unfortunately, the way this issue was addressed in \cite{nach2009} is inadequate, as was pointed out to us by Hendrik Grundling in a series of emails.

The purpose of this note is to remove the shortcomings of \cite{nach2009}. 
With the complete proofs provided here, all our results hold as stated in the original publications.

The points we need to address are of two types. The first type concerns issue with the calculus of functions with values in the bounded linear operators on a separable Hilbert space and the solution of evolution equations with unbounded generators defined on a dense domain. As Grundling pointed out to us, and as we explain in Section \ref{sec:calculus} below, the differentiability, continuity and measurability properties of the functions we encounter are not always self-evident and this makes the generalization of results in finite dimension sometimes less than straightforward. E.g., the norm of a bounded strongly continuous operator-valued function is not necessarily continuous. In general, such a function is measurable only if the underlying Hilbert space is separable. Since it does not appear to be widely known how to deal with such issues, we include the necessary arguments in Section \ref{sec:calculus}.

A second set of issues occurs in the specific application to the dynamics of lattice systems and to explain how to address those we need to recall the setup of \cite{nach2009} (see Section \ref{sec:lrb}). An important role is played by observables with finite support and Lieb-Robinson bounds allow one to bound the rate of growth of the support under the Heisenberg dynamics. It is necessary to separate the terms in the Hamiltonian that preserve the support of a given observable from terms that do not and to estimate the effect of the corresponding components on the dynamics accordingly. By paying proper attention to these issues we provide a correct proof of Theorem 2.1 in \cite{nach2010}. 

\section{Evolution equations in $\mathcal{B}(\mathcal{H})$}\label{sec:calculus}

In this section we review some basic properties of operator-valued functions and prove two preliminary results: 
Proposition~\ref{prop:sols} on the solution of the Schr\"odinger equation with time-dependent 
bounded Hamiltonians, and Lemma~\ref{lem:normbd} on estimating the growth in norm of solutions of a class of
Heisenberg-type equations. 

Let $\mathcal{H}$ be a complex, separable Hilbert space and let $\mathcal{B}( \mathcal{H})$ denote the bounded linear operators on 
$\mathcal{H}$.  A function $A: \mathbb{R} \to \mathcal{B}(\mathcal{H})$ is said to be {\it strongly continuous} (resp. {\it strongly differentiable}) if:
for all $\psi \in \mathcal{H}$, $A(t) \psi$ is continuous (resp. differentiable) in $t$ with respect to the norm topology on $\mathcal{H}$. 
By the Uniform Boundedness Principle, if $A$ is strongly continuous, then $A$ is locally bounded, i.e., if 
$I \subset \mathbb{R}$ is compact, then 
\begin{equation}
\sup_{t \in I} \| A(t) \| < \infty
\end{equation}
As the following example shows, the strong continuity of $t\mapsto A(t)$  does not imply that $t\mapsto \Vert A(t)\Vert$ is continuous.
 Consider a sequence of non-zero, orthogonal projections $P_n$ for which 
$\sum_{n=0}^\infty P_n=\idty$, as a series converging in the strong operator topology. 
For any smooth function $f:[0, \infty)\to [0,1]$, such that $f(x)=1$,  for $x\in [0,1]$, and $f(x)=0$ for $x\geq 2$, define $A(t)$ by
\begin{equation}
A(t) = \left\{ \begin{array}{ll} \idty - \sum _{n=0}^\infty f(nt) P_n, & \mbox{if } t>0, \\ 0 & \mbox{otherwise.} \end{array} \right.
\end{equation}
Then, $A(t)$ is strongly continuous, but $\Vert A(t)\Vert$ is not continuous at $t=0$. To see the strong continuity,
consider the estimate
\begin{equation}
\Vert (A(t)-A(t_0))\psi\Vert^2 =\sum_{n=0}^\infty (f(nt)-f(nt_0))^2 \Vert P_n \psi\Vert^2\leq \sum_{n=\lceil t^{-1}\rceil }^{[2t_0^{-1}]}
(f(nt)-f(nt_0))^2 \Vert P_n \psi\Vert^2. 
\end{equation}
Here we assumed that $t>t_0\geq0$. If $t_0>0$, we find
\begin{equation}
 \sum_{n=\lceil t^{-1}\rceil }^{[2t_0^{-1}]}
(f(nt)-f(nt_0))^2 \Vert P_n \psi\Vert^2\leq  \sum_{n=\lceil t^{-1}\rceil }^{[2t_0^{-1}]} (n | t-t_0| \Vert f^\prime\Vert_\infty)^2 \Vert P_n\psi\Vert^2
\leq 4 t_0^{-2} |t-t_0|^2 \Vert f^\prime\Vert^2_\infty \Vert \psi\Vert^2.
\end{equation}
This shows strong continuity on $(0,\infty)$. In the case $t_0=0$, we have
\begin{equation}
 \sum_{n=\lceil t^{-1}\rceil }^{[2t_0^{-1}]}
(f(nt)-f(nt_0))^2 \Vert P_n \psi\Vert^2\leq  \sum_{n=\lceil t^{-1}\rceil }^{\infty} \Vert P_n\psi\Vert^2.
\end{equation}
and strong continuity follows as the RHS is the tail of a convergent series. Since $\Vert A(0)\Vert =0$ and $\Vert A(t)\Vert =1$ for $t>0$, 
$\Vert A(t)\Vert$ is clearly not continuous at $0$.

Using the local boundedness  it is easy to show that the product of two strongly continuous mappings is jointly strongly 
continuous. In fact, let $A$ and $B$ be strongly continuous mappings. It suffices to observe that, for any $t_0, s_0 \in \mathbb{R}$, 
the norm of the quantity
\begin{equation}
(A(t)B(s) - A(t_0) B(s_0) ) \psi =  A(t)(B(s) - B(s_0) ) \psi + (A(t) - A(t_0)) B(s_0) \psi 
\end{equation} 
converges (as $(t,s) \to (t_0,s_0)$) to zero for every $\psi \in \mathcal{H}$. A similar argument shows that the
product of strongly differentiable mappings is strongly differentiable in each variable. 

In the following sections we will make use of various types of integrals of operator-valued functions. Therefore, 
we briefly discuss the measurability of operator-valued functions. Weak integrals are defined as 
follows. If $A: \mathbb{R} \to \mathcal{B}( \mathcal{H})$ is
weakly measurable, i.e. for all $\phi, \psi \in \mathcal{H}$, $t \mapsto \langle \phi, A(t) \psi \rangle$ is measurable, then
for any interval $I \subset \mathbb{R}$ the weak-integral of $A$ over $I$ is defined as the operator $B_I \in \mathcal{B}( \mathcal{H})$ given by the Lebesgue integral
\begin{equation}
\langle \phi, B_I \psi \rangle = \int_I  \langle \phi, A(t) \psi \rangle \, dt \, . 
\end{equation}
For strongly continuous functions, the same integrals can also be interpreted strongly. 
The fundamental theorem of calculus holds in the weak and the strong sense for weak, respectively, strong
integrals. It is useful to note that for a function $A: \mathbb{R} \to \mathcal{B}(\mathcal{H})$, weak measurability implies that 
the functions $t\mapsto \Vert A(t)\psi\Vert$, for any $\psi\in \mathcal{H}$, and $t\mapsto \Vert A(t)\Vert$, are measurable. This follows 
from the existence of a countable dense set, $\mathcal{S}_0$, in the unit sphere in $\mathcal{H}$ (it is here that we use that 
$\mathcal{H}$ is separable) and the fact that the supremum of a countable set of measurable functions is measurable 
(see, e.g.,  \cite[Proposition 2.7] {Foll}). Explicitly, if for all $\phi,\psi\in\mathcal{H}$, $t\mapsto \langle \phi, A(t)\psi\rangle$ is measurable, then so is $t\mapsto \vert\langle \phi, A(t)\psi\rangle\vert$, and therefore also the following two functions defined as
a supremum of measurable functions:
\begin{eqnarray}
\| A(t) \psi\| &=& \sup\{ \vert\langle\phi, A(t) \psi \rangle\vert\mid  \phi\in\mathcal{S}_0\}\\
\| A(t) \| &=& \sup\{ \| A(t) \psi \|\mid  \psi\in\mathcal{S}_0\}.
\end{eqnarray}
As a consequence, the bound
\begin{equation}
| \langle \phi, B_I \psi \rangle | \leq \int_I | \langle \phi, A(t) \psi \rangle | dt \leq \| \phi \| \int_I \| A(t) \psi \| dt 
\end{equation}
immediately yields
\begin{equation}\label{int1}
\left\| \left(\int_I A(t)dt \right) \psi \right\| \leq \int_I \| A(t) \psi\|  dt
\end{equation}
where we have used the more common notation $\int_I A(t)dt$ for $B_I$.
A similar argument shows that
\begin{equation}\label{int2}
\left\| \int_I A(t) dt \right\| \leq \int_I \| A(t) \|  dt.
\end{equation}
The integrals we estimated in \eq{int1} and \eq{int2} are defined as weak integrals.
If the integrals also exist in a stronger sense, clearly, the same inequalities hold. 
 
The Schr\"odinger dynamics generated by a bounded, time-dependent Hamiltonian is 
well understood and solutions are 
often expressed by a Dyson series. A standard result in this direction assumes that 
$H:\mathbb{R}\to\mathcal{B}(\mathcal{H})$ is strongly continuous and self-adjoint, i.e. $H(t)^*= H(t)$ for all $t \in \mathbb{R}$
(see, e.g.,Theorem X.69 of \cite{RS2}). Under this assumption,
for every initial condition $\psi_0\in\mathcal{H}$,  the time-dependent Schr\"odinger equation
\begin{equation} \label{tdseqn}
i\frac{d}{dt}\psi(t)=H(t)\psi(t), \quad \psi(0)=\psi_0,
\end{equation}
has a unique solution in the sense that there is a unique differentiable function $\psi(t)$ which satisfies (\ref{tdseqn}). This solution is easily seen to give rise to a unitary propagator $U(t,s)\in \mathcal{B}(\mathcal{H})$ such that 
\begin{equation}
\psi(t)=U(t,s)\psi(s),\quad t,s\in\mathbb{R}
\end{equation}
with $U(t,s)$ separately, strongly differentiable in both $t$ and $s$.
An explicit construction of this propagator is given by the Dyson series:
\begin{equation}\label{dyson}
U(t,s)\psi =  \psi + \sum_{n=1}^{\infty} (-i)^n \int_{s}^t \int_{s}^{t_1} \cdots \int_{s}^{t_{n-1}} H(t_1) \cdots H(t_n)  \psi \, dt_n \cdots dt_1\,
\end{equation}
for any $\psi \in \mathcal{H}$. It is easily seen that the propagator $U(t,s)$ itself satisfies the equation
\begin{equation}\label{se_propagator}
\frac{d}{dt}U(t,s) =-iH(t)U(t,s),
\end{equation}
which holds in the strong sense.

The additional observation we want to make here is that $U(t,s)$ is not only the unique strong solution of this equation. It is also a fact that
any bounded weak solution of (\ref{se_propagator}) necessarily coincides with $U(t,s)$. By {\em weak solution}, we mean that for all $\phi,\psi\in\mathcal{H}$,
\begin{equation}
 \frac{d}{dt}\langle\phi,U(t,s)\psi\rangle =-i\langle\phi, H(t)U(t,s)\psi\rangle.
\end{equation}
This also implies that if, for strongly continuous, self-adjoint, and bounded $H(t)$, the operators $U(t,s)$ satisfy
\begin{equation}\label{se_propagator2}
\frac{d}{dt}U(t,s) \psi=-iH(t)U(t,s)\psi,
\end{equation}
for $\psi$ in a dense subset of $\mathcal{H}$, then the equation is satisfied for all $\psi\in\mathcal{H}$, that $U(t,s)$ is separately, strongly differentiable in both $t$ and $s$, and that $U(t,s)$ is given by the Dyson series (\ref{dyson}).
This is the content of the following proposition.

\begin{prop} \label{prop:sols}
Let $A:  \mathbb{R} \to \mathcal{B}(\mathcal{H})$ be strongly continuous. Then the following statements hold:

(i) The equation 
\begin{equation}\label{basic_de}
\frac{d}{dt} V(t)= A(t) V(t), \quad V(t_0) = V_0 \in \mathcal{B}( \mathcal{H})
\end{equation}
has a unique strong solution $V : \mathbb{R} \to \mathcal{B}( \mathcal{H})$.

(ii) Any locally norm-bounded, weak solution coincides with the strong solution.

(iii) If $V: \mathbb{R} \to \mathcal{B}(\mathcal{H})$ is strongly continuous and satisfies
\begin{equation}\label{de_on_vectors}
\frac{d}{dt} V(t)\psi= A(t) V(t)\psi,  \quad V(t_0) \psi= V_0\psi
\end{equation}
for $\psi$ in a dense subset of $\mathcal{H}$, then $V(t)$ satisfies (\ref{de_on_vectors}) for all $\psi\in\mathcal{H}$. In other words, $V(t)$ is the strong solution.

(iv) If $V_0$ is invertible, the solution $V(t)$ is invertible for all $t\in\mathbb{R}$, and the inverse $V(t)^{-1}$ is strongly differentiable.

(v) If $A(t)=-iH(t)$, with $H(t)$ bounded and self-adjoint for all $t\in\mathbb{R}$, the solution $V(t)$ with initial condition 
$V_0=\idty$ is unitary for all $t\in\mathbb{R}$ and, using the uniqueness, it is easy to verify that $U(t,s)=V(t)V(s)^{-1}$ is a 
unitary cocycle also given by (\ref{dyson}).
\end{prop}

\begin{proof}
(i)
The solution can be constructed using the Dyson series (\ref{dyson}) in the same way as is done for self-adjoint $A(t)$. 
In fact, for any $t \in \mathbb{R}$ and $\psi \in \mathcal{H}$, consider
\begin{equation}\label{dyson2}
V(t)\psi =  V_0\psi + \sum_{n=1}^{\infty}  \int_{t_0}^t \int_{t_0}^{t_1} \cdots \int_{t_0}^{t_{n-1}} A(t_1) \cdots A(t_n)  V_0\psi \, dt_n \cdots dt_1\,
\end{equation}
The integrals appearing as terms in the series above are well-defined due to the strong continuity of 
$A$; in particular, for any $n \geq 1$, products of the form $A(t_1) \cdots A(t_n)$
are jointly, strongly continuous in the variables $t_1, \cdots, t_n$. Thus, the integrands are 
measurable. Moreover, local boundedness further implies that the series is absolutely convergent, and hence, the 
Dyson series is a strong solution of (\ref{basic_de}). For $t<t_0$, the expression (\ref{dyson2}) may also be written as
\begin{equation}\label{dyson3}
V(t)\psi =  V_0\psi + \sum_{n=1}^{\infty} (-1)^n \int_{t}^{t_0} \int_{t_1}^{t_0} \cdots \int_{t_{n-1}}^{t_0} A(t_1) \cdots A(t_n)  V_0\psi \, dt_n \cdots dt_1\,
\end{equation}

Uniqueness of the strong solution can be proved using Gronwall's Lemma. 
For example, suppose $V_1$ and $V_2$ are two strong solutions. Then, for all $\psi\in\mathcal{H}$ and any $t \in \mathbb{R}$, we have
\begin{eqnarray}
\Vert (V_1(t) - V_2(t)) \psi \Vert & = & \left\Vert \int_{t_0}^{t} \frac{d}{ds}   (V_1(s) - V_2(s)) \psi \, ds \right\Vert  \nonumber \\ 
& = &  \left\Vert \int_{t_0}^{t} A(s)(V_1(t) - V_2(t)) \psi\, ds\right\Vert  \nonumber \\
& \leq & M \int_{t_-}^{t_+} \Vert (V_1(s) - V_2(s)) \psi\Vert \, ds \nonumber
\end{eqnarray}
where we have set $t_+ = \max\{ t, t_0\}$ and $t_- = \min \{t,t_0 \}$.
Gronwall's Lemma then implies that  $(V_1(t) - V_2(t)) \psi =0$. As $\psi$ is arbitrary, this proves the uniqueness of the strong solution.

(ii) The uniqueness of bounded weak solutions can be proved in the same way.
Suppose that $V_1(t)$ and $V_2(t)$ are two weak solutions of (\ref{basic_de}). This means that    
for all $\phi, \psi \in \mathcal{H}$ and any $t \in \mathbb{R}$,
\begin{equation}\label{basic_de2}
\frac{d}{dt} \langle \phi, V_i(t)\psi\rangle = \langle \phi, A(t) V_i(t) \psi \rangle, \quad \mbox{with } \langle\phi,V_i(t_0)\psi\rangle = \langle\phi,V_0\psi\rangle, \quad  \mbox{for } i=1,2.
\end{equation}
Therefore,
\begin{eqnarray}
\left| \langle \phi, \left( V_1(t) - V_2(t) \right) \psi \rangle \right| & = & \left| \int_{t_0}^{t} \frac{d}{ds} \langle \phi, \left( V_1(s) - V_2(s) \right) \psi \rangle \, ds \right|  \nonumber \\ & = &  \left| \int_{t_0}^{t} \langle \phi, A(s)( V_1(s) - V_2(s))\psi \rangle \, ds \right|  \nonumber \\
& \leq & \| \phi \| \| \psi \| M \int_{t_-}^{t_+} \| V_1(s) - V_2(s) \| \, ds 
\end{eqnarray}
where we have, again, set $t_+ = \max \{ t, t_0 \}$ and $t_- = \min \{ t,t_0\}$.
By taking the supremum over all normalized $\phi, \psi \in \mathcal{H}$, this
shows that
\begin{equation}
\| V_1(t) - V_2(t) \| \leq M \int_{t_-}^{t_+} \| V_1(s) - V_2(s) \| \, ds 
\end{equation}
which, upon application of Gronwall's Lemma, implies that $V_1(t) = V_2(t)$ for all $t \in \mathbb{R}$. Since any strong solution is also a weak solution, the unique weak and unique strong solution must coincide.

(iii) Under the assumptions,  there is a dense subset $\mathcal{D}\subset\mathcal{H}$, such that for all $\psi\in\mathcal{D}$ and $\phi\in\mathcal{H}$, we have
\begin{equation}\label{dense_strong}
\frac{d}{dt} \langle\phi,V(t)\psi\rangle= \langle\phi,A(t) V(t)\psi\rangle.
\end{equation}
It is clear that $\Vert V(t)\Vert$ is locally bounded. In fact, the usual proof of Gronwall's Lemma yields an explicit bound in 
terms of $\Vert V_0\Vert$ and the local norm bound on $A$. For $\psi\in\mathcal{H}$, let $\{ \psi_n \}$ be a sequence in $\mathcal{D}$ converging to $\psi$, and consider the sequence of functions
$f_n$ defined by
\begin{equation}
f_n(t)= \langle\phi,V(t)\psi_n\rangle.
\end{equation}
Due to (\ref{dense_strong}), we have $f_n^\prime(t) =  \langle\phi,A(t) V(t)\psi_n\rangle$. Since $A(t)V(t)$ is bounded,
it is easy to see that the sequence $\{ f_n^\prime \}$ converges to $g(t)= \langle\phi,A(t)V(t)\psi\rangle$, and the convergence is uniform on compact sets. Therefore the
conditions of \cite[Theorem 7.17]{Rudin_principles} are satisfied, implying that $f_n(t)$ converges to $f(t)= \langle\phi,V(t)\psi\rangle$
and $f^\prime(t)=g(t)$. This proves that
\begin{equation}
\frac{d}{dt} V(t)\psi= A(t) V(t)\psi\, ,
\end{equation}
for all $\psi\in\mathcal{H}$, and therefore that $V(t)$ is again the unique strong solution of (\ref{basic_de}).

(iv) Consider the solution $W(t)$ of
\begin{equation}
\frac{d}{dt} W(t)= -W(t)A(t), \quad W(t_0) = V_0^{-1},
\end{equation}
which is given by
\begin{equation}\label{dyson_inverse}
W(t)\psi =  V_0^{-1}\psi + \sum_{n=1}^{\infty}  (-1)^n\int_{t_0}^t \int_{t_0}^{t_1} \cdots \int_{t_0}^{t_{n-1}} V_0^{-1} A(t_n) \cdots A(t_1)  
\psi \, dt_n \cdots dt_1\,.
\end{equation}
It is then straightforward to check that $Y(t)=V(t)W(t)$, with $V$ the solution of (\ref{basic_de}), satisfies
\begin{equation}
\frac{d}{dt} Y(t)= A(t)Y(t)-Y(t)A(t), \quad Y(t_0) = \idty,
\end{equation}
which has the unique, constant solution $Y(t) =\idty$. Therefore, we have $W(t)=V(t)^{-1}$.

(v) This follows in a straightforward manner from the arguments given above.
\end{proof}

The above results using the Dyson series also enable a simple norm bound for solutions of certain differential equations.
This enters our proof of the Lieb-Robinson bound, and we state it here as a basic lemma.

\begin{lem} \label{lem:normbd}  Let $A,B: \mathbb{R} \to \mathcal{B}(\mathcal{H})$, 
be strongly continuous, with $A$ self-adjoint, i.e. $A(t)^*=A(t)$, for all $t\in \mathbb{R}$.
Then, for any $t_0 \in \mathbb{R}$, there exists a unique strong solution of the initial value problem 
\begin{equation} \label{gende}
f'(t) = i[A(t), f(t)] + B(t) \quad \mbox{with} \quad f(t_0) = f_0 \in \mathcal{B}(\mathcal{H}).
\end{equation}
This solution $f(t)$ satisfies the estimate
\begin{equation} \label{normbd0}
\| f(t) \| \leq \| f(t_0) \| + \int_{t_-}^{t_+} \| B(s) \| \, ds 
\end{equation}
for $t \in \mathbb{R}$; here we have set $t_+ = \max \{ t, t_0 \}$ and $t_- = \min \{ t,t_0 \}$.
Moreover, any bounded weak solution coincides with the strong solution and, therefore, satisfies the norm bound (\ref{normbd0}).
\end{lem}

\begin{proof} 
With $t_0 \in \mathbb{R}$ fixed, the Dyson series given by 
\begin{equation}
U(t, t_0)  = \idty + \sum_{n=1}^{\infty} i^n \int_{t_0}^t \int_{t_0}^{t_1} \cdots \int_{t_0}^{t_{n-1}} A(t_1) \cdots A(t_n)  \, dt_n \cdots dt_1
\end{equation}
is a well-defined, strongly differentiable family of unitaries satisfying the time-dependent Schr\"odinger equation corresponding to $-A$, i.e. for any $\psi \in \mathcal{H}$,
\begin{equation}
i\frac{d}{dt}U(t, t_0) \psi=-A(t)U(t, t_0) \psi, \quad \mbox{with } U(t_0, t_0) \psi = \psi. 
\end{equation}
One readily checks that $U(t,t_0)^* = U(t_0,t)$. Thus it is also strongly differentiable and 
satisfies
\begin{equation}
i\frac{d}{dt}U(t,t_0)^* \psi=U(t, t_0)^*A(t) \psi, \quad \mbox{with } U(t_0, t_0)^* \psi = \psi. 
\end{equation}
As a consequence, for any $g_0 \in \mathcal{B}(\mathcal{H})$, the mapping $g : \mathbb{R} \to \mathcal{B}( \mathcal{H})$ 
given by 
\begin{equation} \label{stsol}
g(t) = U(t, t_0) g_0 U(t,t_0)^*
\end{equation}
is the unique, strong solution of the initial value problem 
\begin{equation} \label{npde}
g'(t) = i [ A(t), g(t) ] \quad \mbox{with} \quad g(t_0) = g_0 \, .
\end{equation}
Here the uniqueness statement is proven by arguments using Gronwall's Lemma as is done
in Proposition~\ref{prop:sols}.

We now claim that $f : \mathbb{R} \to \mathcal{B}( \mathcal{H})$ given by 
\begin{equation} \label{deff}
f(t)  = U(t, t_0) \left( f_0 + \int_{t_0}^{t} U(s, t_0)^* B(s) U(s, t_0) \, ds \right) U(t, t_0)^*
\end{equation} 
is a strong solution of (\ref{gende}). As a product of strongly differentiable mappings, it is clear that this $f$ is
strongly differentiable. A short calculation shows that this $f$ satisfies (\ref{gende}), and  moreover, 
the bound (\ref{normbd0}) readily follows from (\ref{deff}) and unitarity of $U$. 
Again, as discussed in the proof of Proposition~\ref{prop:sols}, an application of Gronwall's Lemma proves uniqueness (in both the
strong and weak sense), and this completes the proof.
\end{proof}

%
%
%

\section{A proof of the Lieb-Robinson bound}\label{sec:lrb}

The models we consider are defined over a countable set $\Gamma$ equipped with a metric $d$. In the event that
the cardinality of $\Gamma$ is infinite, we will assume that there is a non-increasing function $F: [0, \infty) \to (0, \infty)$ 
for which:

i) $F$ is uniformly integrable, i.e.
\begin{equation} \label{Fint}
\| F \| = \sup_{x \in \Gamma} \sum_{y \in \Gamma} F(d(x,y)) < \infty,
\end{equation}
and

ii) $F$ satisfies the convolution condition
\begin{equation} \label{Fconv}
C = \sup_{x,y \in \Gamma} \sum_{z \in \Gamma} \frac{F(d(x,z))F(d(z,y))}{F(d(x,y))} < \infty
\end{equation}

A quantum system over $\Gamma$ is now defined as follows. To each site $x \in \Gamma$, we associate a separable, complex Hilbert
space $\mathcal{H}_x$. By $\mathcal{B}(\mathcal{H}_x)$ we will denote the algebra of all bounded linear operators over $\mathcal{H}_x$.  
For any finite set $\Lambda \subset \Gamma$, the Hilbert space of states and algebra of local observables over $\Lambda$ will be
denoted by
\begin{equation}
\mathcal{H}_{\Lambda} = \bigotimes_{x \in \Lambda} \mathcal{H}_x \quad \mbox{and} \quad \mathcal{A}_{\Lambda} = 
\mathcal{B}( \mathcal{H}_\Lambda) 
\end{equation}
For any two finite sets $\Lambda_0 \subset \Lambda \subset \Gamma$, $\mathcal{A}_{\Lambda_0}$ can be naturally identified with
the subset of $\mathcal{A}_{\Lambda}$ consisting of $\tilde{A} = A \otimes \idty_{\Lambda \setminus \Lambda_0} \in \mathcal{A}_{\Lambda}$,
for all $A \in \mathcal{A}_{\Lambda_0}$.
The algebra of {\em local observables} is given by the inductive limit
\begin{equation}
\mathcal{A}_{\Gamma}^{\rm loc} = \bigcup_{\Lambda \subset \Gamma} \mathcal{A}_{\Lambda},
\end{equation}
where the union taken over all {\em finite} subsets of $\Gamma$. The completion of $\mathcal{A}_{\Gamma}^{\rm loc}$ with respect to the operator norm, which we denote by $\mathcal{A}_{\Gamma}$,
is a $C^*$-algebra, and it will be called the algebra of all {\em quasi-local observables}. 

The models we will be considering are defined by families of Hamiltonians comprised of two types of terms: strictly local terms and bounded interactions. For each $x \in \Gamma$, there is a self-adjoint operator $H_x$ with dense domain $\mathcal{D}_x \subset\mathcal{H}_x$. 
The bounded interactions are described by a map $\Phi$ from the set of finite subsets of $\Gamma$ to $\mathcal{A}_{\Gamma}^{\rm loc}$ 
with the property that: for each $X \subset \Gamma$ finite, $\Phi(X)^* = \Phi(X) \in \mathcal{A}_{X}$. Then, for each finite $\Lambda \subset \Gamma$, the Hamiltonian for the system in $\Lambda$ is given by
\begin{equation} \label{fvham}
H_{\Lambda} = \sum_{x \in \Lambda}H_x + \sum_{X \subset \Lambda} \Phi(X),
\end{equation}
which is well-defined and essentially self-adjoint on the dense domain (see, e.g., \cite[Theorem VIII.33]{RS1})
\begin{equation} \label{fvdom}
\mathcal{D}_{\Lambda} = \span \{\bigotimes_{x\in\Lambda} \psi_x \mid \psi_x\in \mathcal{D}_x, \mbox{ for all } x\in \Lambda\}.
\end{equation}
 Using the spectral theorem, one can define the Heisenberg dynamics, $\tau_t^{\Lambda}$, generated
 by this self-adjoint operator, which is the one parameter group of automorphisms of $\mathcal{A}_{\Lambda}$ defined by
 \begin{equation} \label{dynamics}
 \tau_t^{\Lambda}(A) = e^{i t H_{\Lambda}} A e^{-itH_{\Lambda}} \quad \mbox{for any} \quad A \in \mathcal{A}_{\Lambda}.
 \end{equation}

By Stone's theorem, see e.g. Section VIII.4 of \cite{RS1} or Theorem 7.3.7 of \cite{Weid}, the unitaries $t\mapsto U_t^{\Lambda}= e^{-itH_{\Lambda}}$ are strongly continuous, leave the domain of $H_{\Lambda}$ invariant, and satisfy 
 \begin{equation} \label{stone}
 \frac{d}{dt} U_t^{\Lambda} \psi = -iH_{\Lambda} U_t^{\Lambda} \psi = -i U_t^{\Lambda} H_{\Lambda} \psi \quad \mbox{for all } \quad \psi \in \mathcal{D}_{\Lambda}.
 \end{equation}
We conclude that for any  $A \in \mathcal{A}_{\Lambda}$, the
 time evolution of $A$, defines a strongly continuous function $t\mapsto A(t) = \tau_t^{\Lambda}(A)$, in the sense of Section \ref{sec:calculus}.
 
Lieb-Robinson bounds provide an estimate of the rate at which the support of an observable grows as it evolves under the dynamics
\eq{dynamics}. We will prove a Lieb-Robinson bound for a class of sufficiently short- range interactions defined as follows. 
Let $\Gamma$ and $F$ be taken as above.  The space of interactions $B_F( \Gamma)$ consists of those $\Phi$ for which
\begin{equation} \label{intbd}
\| \Phi \| = \sup_{x,y \in \Gamma} \frac{1}{F(d(x,y))} \sum_{\stackrel{X \subset \Gamma:}{x,y \in X}} \| \Phi(X) \| < \infty
\end{equation} 
For any finite $X \subset \Lambda\subset\Gamma$, we define 
\begin{equation}
S_{\Lambda}(X) = \{ Z \subset \Lambda : Z \cap X \neq \emptyset \mbox{ and } Z \cap (\Lambda \setminus X) \neq \emptyset \},
\end{equation}
the {\em surface} of $X$ in $\Lambda$ and set $S(X) = S_{\Gamma}(X)$. The {\em $\Phi$-boundary} of a set $X$ is then given by
\begin{equation}
\partial_{\Phi} X = \{ x \in X : \exists Z \in S(X) \mbox{ with } x \in Z \mbox{ and } \Phi(Z) \neq 0 \}.
\end{equation} 
For generic $\Phi$, $\partial_{\Phi} X=X$, but if $\Phi$ is of finite range, $\partial_{\Phi} X$ is a strict subset of $X$ when $X$ is sufficiently large. 
We can now state the bound.
 
 \begin{thm} \label{thm:lrb} Let $\Gamma$ and $F$ be as indicated above. Fix a collection of local Hamiltonians $\{ H_x \}_{x \in \Gamma}$ and an interaction 
 $\Phi \in B_F( \Gamma)$. Let $X,Y \subset \Gamma$ be finite disjoint sets. For any finite $\Lambda \subset \Gamma$ with 
 $X \cup Y \subset \Lambda$ and any $A \in \mathcal{A}_X$ and $B \in \mathcal{A}_Y$, the bound
 \begin{equation} \label{lrbest}
 \left\| \left[ \tau_t^{\Lambda}(A), B \right] \right\| \leq \frac{2 \| A \| \| B \|}{C} (e^{2 \| \Phi \| C |t|} - 1) D(X,Y)
 \end{equation}
 holds for all $t \in \mathbb{R}$, where the quantity $D(X,Y)$ is given by
 \begin{equation} \label{lrbmin}
 D(X,Y) = \min \left\{ \sum_{x \in X} \sum_{y \in \partial_{\Phi} Y} F(d(x,y)), \sum_{x \in \partial_{\Phi} X} \sum_{y \in Y} F(d(x,y))\right\}
 \end{equation}
 \end{thm}
 Note that if $X \cap Y \neq \emptyset$, one always has $ \left\| \left[ \tau_t^{\Lambda}(A), B \right] \right\| \leq 2 \| A \| \| B \|$. If $\Phi$ is 
 exponentially decaying in the sense that there exists $a>0$ such that $\Phi\in B(\Gamma, F_a)$, with $F_a(r)=e^{-ar}F(r)$, $F(r)$ as above, 
 then
 \begin{eqnarray}
 D(X,Y)
 &\leq& \min \left\{ \sum_{x \in X} \sum_{y \in \partial_{\Phi} Y} F(d(x,y)), \sum_{x \in \partial_{\Phi} X} \sum_{y \in Y} F(d(x,y))\right\} e^{-ad(X,Y)}
 \nonumber\\
 &\leq& \min \{\vert \partial_{\Phi} X\vert,\vert \partial_{\Phi} Y\vert\} \Vert F\Vert e^{-ad(X,Y)},
 \end{eqnarray}
 and the upper bound \eq{lrbest} can be replaced by one of the exponential form found in \cite{LR1}.
 
 \begin{proof}
 We prove Theorem~\ref{thm:lrb} in four steps. First, we define an interaction-picture dynamics; this is (\ref{intdyn}) below.
 For this dynamics, one can check strong differentiability; this is step 1. Next, we show that Lemma~\ref{lem:normbd} applies
 and prove a basic bound, see (\ref{fbd}), for this dynamics. In the third step, we argue that an analogous bound also holds for the
 full dynamics under consideration, this is (\ref{iterandum}). Finally, the desired bound, i.e. (\ref{lrbest}) above, follows from iteration of (\ref{iterandum}).  
 
 {\it Step 1:} 
 Fix any finite subsets $X \subset \Lambda$ of $\Gamma$, and consider the Hamiltonian
 \begin{equation} \label{backham}
 H_0 = \sum_{x \in \Lambda}H_x + \sum_{Z \subset X} \Phi(Z) 
 \end{equation}
 which contains all strictly local Hamiltonians $H_x$ with $x \in \Lambda$, but only those interaction 
 terms with support strictly contained within $X$. It is clear that $H_0$ is self-adjoint on 
 the same dense domain $\mathcal{D}_{\Lambda}$ as the full Hamiltonian $H_{\Lambda}$, see
(\ref{fvham}) and (\ref{fvdom}). To ease notation in what follows, we will drop the dependence of $H_{\Lambda}$ on $\Lambda$ and just write $H$.
 
The following defines a two-parameter family of unitaries on $\mathcal{H}_{\Lambda}$
\begin{equation}
 W(t,s)  = e^{it H_0} e^{-i(t-s)H} e^{-isH_0} \quad \mbox{for any } s,t \in \mathbb{R}.
 \end{equation}
 One readily checks that $W(s,s) = \idty$ for any $s \in \mathbb{R}$, and moreover, 
 \begin{equation}
 W(t,s)^* = W(s,t) = W(t,s)^{-1} \quad \mbox{for any } s,t \in \mathbb{R}.
 \end{equation}
 As the product of strongly continuous, bounded functions, $W(t,s)$ is 
 jointly strongly continuous in $s$ and $t$. The interaction-picture dynamics is defined as
  \begin{equation} \label{intdyn}
 \tau_t^{\rm int}(A) = W(0,t)AW(t,0) = e^{itH} e^{-itH_0}Ae^{itH_0}e^{-itH} \quad \mbox{for any } A \in \mathcal{A}_{\Lambda}
 \end{equation}
We will now show that $t\mapsto \tau^{\rm int}_t(A)$ is strongly differentiable.
 
As we have noted, both $H$ and $H_0$ are well-defined self-adjoint operators with the same dense domain $\mathcal{D}_\Lambda \subset \mathcal{H}_{\Lambda}$.
Stone's theorem, as applied to $H_0$, provides an analogue of (\ref{stone}) for the strongly continuous unitaries 
$U_t^0 = e^{-itH_0}$ which is valid for all $\psi \in \mathcal{D}_{\Lambda}$. A calculation shows that for any $\psi \in \mathcal{D}_{\Lambda}$
 \begin{equation}
 \frac{d}{dt}W(t,s) \psi = -i e^{itH_0}(H-H_0)e^{-i(t-s)H} e^{-isH_0} \psi = -i H_{\rm int}(t) W(t,s) \psi,
 \end{equation}
where we have written
\begin{equation} \label{inthamW}
H_{\rm int}(t) = e^{itH_0}(H-H_0) e^{-itH_0} = \sum_{\stackrel{Z \subset \Lambda:}{Z \cap (\Lambda \setminus X) \neq \emptyset}} e^{itH_0} \Phi(Z) e^{-itH_0}.
\end{equation}
Since the operator $H_{\rm int}(t)$ is strongly continuous and self-adjoint, Proposition~\ref{prop:sols} (iii) implies that 
$W(t,s)$ is strongly differentiable in $t$. In fact, 
we conclude that $W(t,s)$ is separately strongly differentiable in $t$ and $s$ with strong derivatives given by
\begin{equation} \label{Wder}
 \frac{d}{dt}W(t,s) = -i H_{\rm int}(t) W(t,s) \quad \mbox{and} \quad  \frac{d}{ds}W(t,s) = i W(t,s) H_{\rm int}(s) \, .
\end{equation}
As a consequence, the interaction dynamics is strongly differentiable with derivative
 \begin{equation}
 \frac{d}{dt} \tau_t^{\rm int}(A) = i \tau_t^{\rm int} \left( \left[ H_{\rm int}(t), A \right] \right).
 \end{equation}
 Finally, we observe that for any $A \in \mathcal{A}_X$, with $X$ as in the definition of $H_0$, we have that
 \begin{equation}
  \left[ H_{\rm int}(t), A \right] =  \left[ \tilde{H}_{\rm int}(t), A \right] \quad \mbox{where} \quad \tilde{H}_{\rm int}(t) = \sum_{Z \in S_{\Lambda}(X)} e^{itH_0} \Phi(Z) e^{-itH_0}.
 \end{equation} 

{\it Step 2:} Fix $A \in \mathcal{A}_X$, $B \in \mathcal{A}_Y$, and $\Lambda$ with $X \cup Y \subset \Lambda$. In this step, the sets $X$ and $Y$ need not be disjoint. Consider the function
 \begin{equation} \label{deffapp}
  f(t) = \left[ \tau_t^{\rm int}(A), B \right]
 \end{equation}
 with $\tau_t^{\rm int} $ defined with respect to $X \subset \Lambda$ as in Step 1 above.
 $f$ is strongly differentiable and the strong derivative of $f$ satisfies
 \begin{eqnarray} \label{fde}
 f'(t) & = & i \left[ \tau_t^{\rm int} \left( [ \tilde{H}_{\rm int}(t), A ] \right), B \right]  \nonumber \\
 & = & i \left[ \tau_t^{\rm int} \left( \tilde{H}_{\rm int}(t) \right), f(t) \right] + i \left[ \tau_t^{\rm int}(A), \left[ B, \tau_t^{\rm int} \left( \tilde{H}_{\rm int}(t) \right)  \right] \right]  
 \end{eqnarray}
 where, for the last equality, we used the Jacobi identity. 
 
 We will now apply Lemma~\ref{lem:normbd} with $t_0 = 0$ and the following definitions of $A(t)$ and $B(t)$:
 \begin{equation}
A(t) = \tau_t^{\rm int} \left( \tilde{H}_{\rm int}(t) \right)  = \sum_{Z \in S_{\Lambda}(X)} \tau_t( \Phi(Z) )
 \end{equation}
 \begin{equation}
 B(t) =  i \left[ \tau_t^{\rm int}(A), \left[ B, \tau_t^{\rm int} \left( \tilde{H}_{\rm int}(t) \right)  \right] \right] 
 \end{equation}
 $A(t)$ and $B(t)$ are strongly continuous and $A(t)$ is self-adjoint. Therefore, Lemma~\ref{lem:normbd} applies and
hence $f$  satisfies the estimate
 \begin{equation} \label{fbd}
 \| f(t) \| \leq \| f(0) \| + 2 \| A \| \sum_{Z \in S_{\Lambda}(X)} \int_{t_-}^{t_+} \| [ \tau_s( \Phi(Z)), B ] \| \, ds 
 \end{equation}
 with $t_- = \min\{ t,0\}$ and $t_+ = \max\{t,0\}$. We will also use the trivial estimate
\begin{equation} \label{fat0}
\Vert f(0)\Vert =\Vert [A,B]\Vert \leq 2\Vert A\Vert \Vert B\Vert \delta_Y(X)
\end{equation}
where, for any $Y \subset \Gamma$, the function $\delta_Y$ is defined by
\begin{equation}
\delta_Y(X) = \left\{\begin{array}{ll} 1 & \mbox{if } X\cap Y\neq\emptyset \\ 0 & \mbox{otherwise.} \end{array} \right.
\end{equation}

{\it Step 3:}
We now consider $A \in \mathcal{A}_X$, $B \in \mathcal{A}_Y$ and $\Lambda$ with $X \cup Y \subset \Lambda$ as in the statement of the
theorem. Introduce the interaction-picture dynamics with respect to $X \subset \Lambda$ as in Step 1.
Note that $\tau_t = \tau^{\rm int}_t \circ \tau^{(0)}_t$ where we have set  $\tau_t^{(0)}(A) = e^{itH_0}Ae^{-itH_0}$. 
Using  (\ref{fbd}) and (\ref{fat0}), we conclude that the linear mapping $G_t:\mathcal{A}_X\to
\mathcal{A}_\Lambda$ with $G_t(A)=f(t)$ has a norm bound of the form
 \begin{equation} \label{Gbd}
 \| G_t  \| \leq 2 \Vert B\| \delta_Y(X)+ 2 \sum_{Z \in S_{\Lambda}(X)}  \int_{t_-}^{t_+} \| [ \tau_s( \Phi(Z)), B ] \| \, ds .
 \end{equation}
Since $\tau^{(0)}_t(A)\in\mathcal{A}_X$ for all $t \in \mathbb{R}$, the norm of $G_t (\tau^{(0)}_t(A))=
[\tau^{\rm int}_t\circ \tau^{(0)}_t (A),B]=[\tau_t(A),B]$ can be bounded as follows
\begin{equation} \label{iterandum}
\Vert [\tau_t(A),B]\Vert \leq 2 \Vert A\Vert \Vert B\Vert \delta_Y(X)
+2\Vert A\Vert \sum_{Z\in S_{\Lambda}(X)} \int_{t_-}^{t_+} \Vert [\tau_s(\Phi(Z)),B]\Vert \, ds.
\end{equation}

{\it Step 4:} Iteration of (\ref{iterandum}) yields
\begin{equation}
\| [\tau_t(A),B] \| \leq 2 \| A \| \| B \| \sum_{n=1}^{\infty} \frac{(2|t|)^n}{n!} a_n 
\end{equation}
where for $n \geq 1$,
\begin{equation}
a_n = \sum_{Z_1 \in S_{\Lambda}(X)} \sum_{Z_2 \in S_{\Lambda}(Z_1)} \cdots \sum_{Z_n \in S_{\Lambda}(Z_{n-1})} \delta_Y(Z_n) \prod_{j=1}^n \| \Phi(Z_j) \|,
\end{equation}
and the term corresponding to $n=0$ vanishes since $X$ and $Y$ are disjoint. Note that for $\Phi \in B( \Gamma, F)$, the following bounds
hold:
\begin{equation}
a_1 \leq \sum_{y \in Y} \sum_{\stackrel{Z \in S(X):}{y \in Z}} \| \Phi(Z) \| \leq \| \Phi \| \sum_{y \in Y} \sum_{x \in \partial_{\Phi} X} F(d(x,y))
\end{equation}
\begin{eqnarray}
a_2 & \leq & \sum_{Z_1 \in S(X)} \| \Phi(Z_1) \| \sum_{y \in Y} \sum_{\stackrel{Z_2 \in S(Z_1):}{y \in Z_2}} \| \Phi(Z_2) \|    \nonumber \\
& \leq & \| \Phi \|  \sum_{Z_1 \in S(X)} \| \Phi(Z_1) \| \sum_{y \in Y} \sum_{z_1 \in \partial_{\Phi}Z_1} F(d(z_1, y)) \nonumber \\
& \leq & \| \Phi \| \sum_{y \in Y} \sum_{z_1 \in \Gamma} F(d(z_1,y)) \sum_{\stackrel{Z_1 \in S(X):}{z_1 \in Z_1}} \| \Phi(Z_1) \| \nonumber \\
& \leq & \| \Phi \|^2 \sum_{y \in Y} \sum_{x \in \partial_{\Phi}X} \sum_{z_1 \in \Gamma} F(d(x, z_1)) F(d(z_1,y)) \nonumber \\
& \leq & C \| \Phi \|^2 \sum_{y \in Y} \sum_{x \in \partial_{\Phi}X} F(d(x,y)),
\end{eqnarray}
and more generally,
\begin{equation}
a_n \leq \| \Phi \|^n C^{n-1} \sum_{y \in Y} \sum_{x \in \partial_{\Phi}X} F(d(x,y)).
\end{equation}
An estimate of the form (\ref{lrbest}) follows. Note that since the dynamics is an automorphism, the same bound holds for 
$\| [\tau_{-t}^{\Lambda}(B), A] \| = \| [ \tau_t^{\Lambda}(A), B] \|$, and hence we can use the minimum in 
(\ref{lrbmin}).  
\end{proof}

%
%
%

\section{On the existence of the thermodynamic limit}

It is well-known, see e.g \cite{bratteli1997}, Lieb-Robinson bounds are useful in proving
the existence of the thermodynamic limit of the dynamics for quantum spin systems. The same is true in this setting. The following 
result is from  \cite{nach2010}.
 
\begin{thm}\label{thm:existbd}
Let $\Gamma$ and $F$ be as described in Section~\ref{sec:lrb}. Fix a collection of on-site Hamiltonians $\{ H_x \}_{x \in \Gamma}$
and an interaction $\Phi \in B_F(\Gamma)$. For each $t \in \mathbb{R}$ and $A \in \mathcal{A}_{\Gamma}^{\rm loc}$, 
the norm limit 
\begin{equation}\label{eq:claim} 
\tau_t(A) = \lim_{\Lambda \to \Gamma} \, \tau_t^{\Lambda} (A)
\end{equation} 
exists and the convergence is uniform for $t$ in compact sets. 
The limit may be taken along any increasing sequence of finite sets $\Lambda$ which tend to $\Gamma$, and the result is 
independent of the particular sequence. This limiting dynamics $\tau_t( \cdot)$ can be uniquely extended to a one-parameter group of 
$*$-automorphisms on $\mathcal{A}_{\Gamma}$. 
\end{thm}

\begin{proof} 
Let $\{ \Lambda_n \}_{n \geq 0}$ be any non-decreasing, exhaustive sequence of finite subsets of $\Gamma$.
Let $A \in \mathcal{A}_{\Gamma}^{\rm loc}$ and denote by $X \subset \Gamma$ the finite support of $A$.
For any $T>0$, we will show that the sequence $\{ \tau_t^{\Lambda_n}(A) \}_{n \geq 0}$ is Cauchy in norm, uniformly for $t \in [-T, T]$.
 
It will again be convenient to define an interaction-picture dynamics. 
In this case, for any finite $\Lambda \subset \Gamma$, define a two-parameter family of 
unitaries on $\mathcal{H}_{\Lambda}$ by setting 
\begin{equation} \label{eq:intuni}
 U_{\Lambda} (t,s) = e^{i t H_{\Lambda}^{\text{loc}} } \, e^{-i (t-s) H_{\Lambda}} \, 
 e^{-is H_{\Lambda}^{\text{loc}}} 
\end{equation}
where $H_{\Lambda}$ is as in (\ref{fvham}) and 
\begin{equation}
H_{\Lambda}^{\rm loc} = \sum_{x \in \Lambda} H_x
\end{equation}
is just the strictly local part of $H_{\Lambda}$. The finite volume interaction-picture dynamics in $\Lambda$ is then defined by
\begin{equation} \label{eq:intpic}
\tau^{\Lambda}_{t, \text{int}} (A) = U_{\Lambda} (0,t) \, A \; U_{\Lambda} (t,0),
\quad \mbox{for all } A \in \mathcal{A}_{\Lambda} \, .
\end{equation} 
Arguing as in Step 1 of the proof of Theorem~\ref{thm:lrb}, it is clear that $U_{\Lambda}(t,s)$ is
separately strongly differentiable in $s$ and $t$ with 
\begin{equation}
 \frac{d}{d t} \, U_{\Lambda} (t,s) =  -i H_{\Lambda}^{\rm int} (t) \, U_{\Lambda} (t,s) 
\quad \mbox{and} \quad
\frac{d}{d s} \, U_{\Lambda} (t,s) =  i U_{\Lambda} (t,s) \, H_{\Lambda}^{\rm int} (s) 
\end{equation}
where we have set
\begin{equation} \label{eq:gen}
H^{\rm int}_{\Lambda} (t) =  \sum_{Z \subset \Lambda} e^{it H_{\Lambda}^{\rm loc} } \, 
\Phi (Z)  \, e^{-it H^{\rm loc}_{\Lambda} } .
\end{equation}

Our first goal is to show that the sequence $\{ \tau_{t, \text{int}}^{\Lambda_n}(A) \}_{n \geq 0}$ is Cauchy in norm, uniformly for $t \in [-T,T]$. 
In fact, we will prove the bound
\begin{equation} \label{norm-bd}
\left\| \tau_{t,\text{int}}^{\Lambda_n} (A) - \tau_{t,\text{int}}^{\Lambda_m} (A) \right\| \leq  2 T(1+ e^{2C \| \Phi \| T}) \| A \| \| \Phi \| \sum_{x \in X} \sum_{y \in \Lambda_n \setminus \Lambda_m} F(d(x,y)) ,
\end{equation}
which is valid for any $A \in \mathcal{A}_X$, $t \in [-T, T]$, $m, n$ large, and  with $X \subset \Lambda_m \subset \Lambda_n$. 
Since $|X|$ is finite and $F$ is uniformly integrable, this bound clearly goes to zero (uniformly for $t \in [-T, T]$) as
$m,n \to \infty$. Moreover, using that 
\begin{equation*}
\tau_t^{\Lambda} (A) = \tau_{t,\text{int}}^{\Lambda} \left(e^{itH_{\Lambda}^{\text{loc}}} \, A \,  
e^{-it H_{\Lambda}^{\text{loc}}} \right) = 
\tau_{t,\text{int}}^{\Lambda} \left(e^{it \sum_{x \in X}H_x} \, A \, e^{-i t \sum_{x \in X} H_x} \right) 
\end{equation*}
one immediately sees that the sequence $\{ \tau_t^{\Lambda_n}(A) \}$ is Cauchy in norm, uniformly for $t \in [-T,T]$. Here we argue as in
Step 3 of Theorem~\ref{thm:lrb}. 

To prove (\ref{norm-bd}), let $A$ and $t$ be as above and take $m \leq n$ large with $X \subset \Lambda_m \subset \Lambda_n$. Observe that
\begin{equation} \label{eq:diff}
\tau_{t,\text{int}}^{\Lambda_n} (A) - \tau_{t,\text{int}}^{\Lambda_m} (A) = \int_0^t \frac{d}{ds} 
\left\{ U_{\Lambda_n} (0,s) \, U_{\Lambda_m} (s,t) \, A \, U_{\Lambda_m} (t,s) \, 
U_{\Lambda_n} (s,0) \right\} \, ds \, .
\end{equation}

A short calculation shows that
\begin{equation}
\begin{split}
\frac{d}{ds} U_{\Lambda_n} (0,s) & \, U_{\Lambda_m} (s,t) \, A \, U_{\Lambda_m} (t,s) \, 
U_{\Lambda_n} (s,0) \\
& = \, i U_{\Lambda_n}(0,s) \left[ \left( H^{\text{int}}_{\Lambda_n}(s) - 
H^{\text{int}}_{\Lambda_m}(s) \right), U_{\Lambda_m}(s,t) \, A \, 
U_{\Lambda_m}(t,s) \right] U_{\Lambda_n}(s,0) \\
& = \, i \sum_{\stackrel{Z \subset \Lambda_n:}{ Z \cap (\Lambda_n \setminus \Lambda_m) \neq 
\emptyset}} U_{\Lambda_n}(0,s) e^{is H_{\Lambda_m}^{\text{loc}}} \left[ \Phi(Z,s), 
\tau_{t-s}^{\Lambda_m} \left( \tilde{A}(t) \right) \right] e^{-is H_{\Lambda_m}^{\text{loc}}} 
U_{\Lambda_n}(s,0) \, ,
\end{split}
\end{equation}
where we have set
\begin{equation} \label{eq:tat}
\tilde{A}(t) = e^{-it H_{\Lambda_m}^{\text{loc}}} A \, e^{it H_{\Lambda_m}^{\text{loc}}} \quad \mbox{and} \quad \Phi(Z,s) = e^{is H_{\Lambda_n \setminus \Lambda_m}^{\text{loc}}} \Phi(Z)  e^{-is H_{\Lambda_n 
\setminus \Lambda_m}^{\text{loc}}} 
\end{equation}
and used the fact that 
\begin{eqnarray} \label{eq:tbs}
e^{-is H_{\Lambda_m}^{\text{loc}}}\left( H^{\text{int}}_{\Lambda_n}(s) - 
H^{\text{int}}_{\Lambda_m}(s) \right) e^{is H_{\Lambda_m}^{\text{loc}}} 
& = & \sum_{Z \subset \Lambda_n} e^{is H_{\Lambda_n \setminus \Lambda_m}^{\text{loc}}} 
\Phi(Z)  e^{-is H_{\Lambda_n \setminus \Lambda_m}^{\text{loc}}} - \sum_{Z \subset \Lambda_m} 
\Phi(Z) \nonumber \\
& = & \sum_{\stackrel{Z \subset \Lambda_n:}{ Z \cap (\Lambda_n \setminus \Lambda_m) \neq 
\emptyset}} \Phi(Z,s).
\end{eqnarray}
As a consequence, we have the norm bound
\begin{equation} \label{base-est}
\left\| \tau_{t,\text{int}}^{\Lambda_n} (A) - \tau_{t,\text{int}}^{\Lambda_m} (A)  \right\| \leq \sum_{\stackrel{Z \subset \Lambda_n:}{ Z \cap (\Lambda_n \setminus \Lambda_m) \neq 
\emptyset}} \int_{t_-}^{t_+} 
\left\| \left[ \tau_{t-s}^{\Lambda_m} \left( \tilde{A}(t) \right), \Phi(Z,s)  \right] \right\| \, ds =  \Sigma_1 + \Sigma_2
\end{equation}
where the terms in  the sum on $Z$ have been separated into two groups; those in $\Sigma_1$ contain all the 
non-trivial interaction terms with support $Z$ satisfying $Z \cap (\Lambda_n \setminus \Lambda_m) \neq \emptyset$ and
$Z \cap X \neq \emptyset$, while the rest are contained in $\Sigma_2$.

A simple, over-counting argument shows that
\begin{eqnarray} \label{sig1bd}
\Sigma_1 & \leq & \sum_{x \in X} \sum_{y \in \Lambda_n \setminus \Lambda_m} \sum_{\stackrel{Z \subset \Lambda_n:}{ x,y \in Z}} \int_{t_-}^{t_+} 
2 \| \tilde{A}(t) \| \| \Phi(Z,s) \| \, ds \nonumber \\
& \leq & 2T \| A \| \| \Phi \| \sum_{x \in X}  \sum_{y \in \Lambda_n \setminus \Lambda_m} F(d(x,y)).
\end{eqnarray}

To estimate $\Sigma_2$, we first use the Lieb-Robinson bound, i.e. (\ref{lrbest}). Observe that the supports of the
observables being considered do not expand, i.e.,
\begin{equation} 
{\rm supp}(\tilde{A}(t)) \subset X \quad \mbox{and similarly} \quad {\rm supp}(\Phi(Z,s)) \subset Z 
\end{equation}
and moreover, for these $Z$, $Z \cap X = \emptyset$. In this case, (\ref{lrbest}) provides a norm bound on the bilinear mapping 
$G_t: \mathcal{A}_X \times \mathcal{A}_Z \to \mathcal{A}_{\Lambda}$ with $G_t(A,B) = [ \tau_t^{\Lambda}(A), B]$. 
We conclude that 
\begin{eqnarray}
\left\| \left[ \tau_{t-s}^{\Lambda_m} \left( \tilde{A}(t) \right), \Phi(Z,s)  \right] \right\| & \leq & \frac{2 \| \tilde{A}(t) \| \| \Phi(Z,s) \|}{C} \left( e^{2 C \| \Phi \| |t-s|} -1 \right) \sum_{x \in X} \sum_{z \in Z} F(d(x,z)) \nonumber \\ 
& \leq & \frac{2 \| A \| \| \Phi(Z) \|}{C} e^{2C \| \Phi \| T} \sum_{x \in X} \sum_{z \in Z} F(d(x,z)).
\end{eqnarray}

Another over-counting argument leads one to
\begin{eqnarray} \label{sig2bd}
\Sigma_2 & \leq & \sum_{y \in \Lambda_n \setminus \Lambda_m}  \sum_{\stackrel{Z \subset \Lambda_n:}{y \in Z}}  \frac{2 \| A \| \| \Phi(Z) \|}{C} e^{2C \| \Phi \| T} \sum_{x \in X} \sum_{z \in Z} F(d(x,z)) |t| \nonumber \\
& \leq & \frac{2 T \| A \|}{C} e^{2C \| \Phi \| T} \sum_{x \in X} \sum_{y \in \Lambda_n \setminus \Lambda_m} \sum_{z \in \Lambda_n} F(d(x,z)) \sum_{\stackrel{Z \subset \Lambda_n:}{y, z \in Z}} \| \Phi(Z) \| \nonumber \\
&  \leq & \frac{2 T \| A \|  \| \Phi \|}{C} e^{2C \| \Phi \| T} \sum_{x \in X} \sum_{y \in \Lambda_n \setminus \Lambda_m} \sum_{z \in \Lambda_n} F(d(x,z)) F(d(z,y)) \nonumber \\
&  \leq & 2 T \| A \|  \| \Phi \| e^{2C \| \Phi \| T} \sum_{x \in X} \sum_{y \in \Lambda_n \setminus \Lambda_m} F(d(x,y)).
\end{eqnarray}

Combining (\ref{sig1bd}) and (\ref{sig2bd}), we find (\ref{norm-bd}) as claimed. The proof of the remaining facts in the statement of this theorem is standard (see, e.g., \cite{Simon1993}).
\end{proof}

\section*{acknowledgement}

We thank Hendrik Grundling for pointing out the errors in our earlier work, and for giving us a chance to address them. This acknowledgement does not imply an endorsement by him of the arguments provided here. The authors benefited from participating in the program on {\em Mathematical Horizons for Quantum Physics 2} at Institute for Mathematical Sciences of National University of Singapore, September 2013. Part of the work reported here was carried out during our visit supported by the Institute. RS 
would also like to acknowledge the hospitality of UC Davis where this work was completed during his sabbatical in academic year 2014-2015.

 \end{document}